\documentclass[reqno,11pt]{amsart}
\usepackage{amsmath, latexsym, amsfonts, amssymb, amsthm, amscd}
\usepackage{graphics,epsf,psfrag}
\numberwithin{equation}{section}

\newtheorem{theorem}{Theorem}[section]
\newtheorem{lemma}[theorem]{Lemma}
\newtheorem{proposition}[theorem]{Proposition}

\newtheorem{remark}[theorem]{Remark}

\newtheorem{assumption}[theorem]{Assumption}

\newcommand{\ind}{\mathbf{1}}
\newcommand{\lip}{\mathrm{lip}}

\newcommand{\N}{\mathbb{N}}

\newcommand{\cA}{{\ensuremath{\mathcal A}} }
\newcommand{\cB}{{\ensuremath{\mathcal B}} }

\newcommand{\cP}{{\ensuremath{\mathcal P}} }

\newcommand{\cL}{{\ensuremath{\mathcal L}} }
\newcommand{\cT}{{\ensuremath{\mathcal T}} }

\newcommand{\bP}{{\ensuremath{\mathbf P}} }
\newcommand{\bE}{{\ensuremath{\mathbf E}} }

\DeclareMathSymbol{\leqslant}{\mathalpha}{AMSa}{"36} 
\DeclareMathSymbol{\geqslant}{\mathalpha}{AMSa}{"3E} 
\DeclareMathSymbol{\eset}{\mathalpha}{AMSb}{"3F}     
\newcommand{\dd}{\,\text{\rm d}}             
 
\newcommand{\bbE}{{\ensuremath{\mathbb E}} }

\newcommand{\bbN}{{\ensuremath{\mathbb N}} }

\newcommand{\bbP}{{\ensuremath{\mathbb P}} }

\newcommand{\bbR}{{\ensuremath{\mathbb R}} }

\newcommand{\bbZ}{{\ensuremath{\mathbb Z}} }

\newcommand{\bt}{{\ensuremath{\mathbf t}} }

\newcommand{\gb}{\beta}

\newcommand{\gep}{\varepsilon}       

\newcommand{\go}{\omega}

\newcommand{\gl}{\lambda}

\makeatletter
\def\captionfont@{\footnotesize}
\def\captionheadfont@{\scshape}

\long\def\@makecaption#1#2{%
  \vspace{2mm}
  \setbox\@tempboxa\vbox{\color@setgroup
    \advance\hsize-6pc\noindent
    \captionfont@\captionheadfont@#1\@xp\@ifnotempty\@xp
        {\@cdr#2\@nil}{.\captionfont@\upshape\enspace#2}%
    \unskip\kern-6pc\par
    \global\setbox\@ne\lastbox\color@endgroup}%
  \ifhbox\@ne 
    \setbox\@ne\hbox{\unhbox\@ne\unskip\unskip\unpenalty\unkern}%
  \fi
  \ifdim\wd\@tempboxa=\z@ 
    \setbox\@ne\hbox to\columnwidth{\hss\kern-6pc\box\@ne\hss}%
  \else 
    \setbox\@ne\vbox{\unvbox\@tempboxa\parskip\z@skip
        \noindent\unhbox\@ne\advance\hsize-6pc\par}%
\fi
  \ifnum\@tempcnta<64 
    \addvspace\abovecaptionskip
    \moveright 3pc\box\@ne
  \else 
    \moveright 3pc\box\@ne
    \nobreak
    \vskip\belowcaptionskip
  \fi
\relax
}
\makeatother
\def\writefig#1 #2 #3 {\rlap{\kern #1 truecm
\raise #2 truecm \hbox{#3}}}

\newcommand{\tf}{\textsc{f}}

\title[Stretched exponential pinning]
{The rounding of the phase transition for disordered pinning with stretched exponential tails}

\address{IMPA - Instituto Nacional de Matem\'atica Pura e Aplicada,
Estrada Dona Castorina 110,
Rio de Janeiro / Brasil 22460-320}

\email{lacoin@ceremade.dauphine.fr}
\author{Hubert Lacoin}

\begin{document}

\maketitle

\begin{abstract}
The presence of frozen-in or quenched disorder in a system can often modify the nature of its phase transition.
A particular instance of this phenomenon is the so-called rounding effect: it has been shown in many cases
that the free-energy curve of the disordered system at its critical point is smoother than that of the homogenous one. In particular some disordered systems do not allow first-order transitions.
We study this phenomenon for the pinning of a renewal with stretched-exponential tails on a defect line (the distribution $K$ of the renewal increments satisfies $K(n) \sim c_K\exp(-n^{\zeta}),$  $\zeta\in (0,1)$) which has a first order transition when disorder is not present.
We show that the critical behavior of the disordered system depends on the value of $\zeta$: when $\zeta>1/2$ the transition remains of first order, whereas the free-energy diagram is smoothed for $\zeta\le 1/2$. Furthermore we show that the rounding effect is getting stronger when $\zeta$ diminishes. \\
{\it Keywords: Disordered pinning, Phase transition, Rounding effect, Harris Criterion.}
\end{abstract}

\section{Introduction}

The effect of a quenched disorder on critical phenomena is a central topic in equilibrium statistical mechanics. In many cases it is expected that the presence of impurities in a system \textit{rounds} or \textit{smoothes} the phase transition in the following sense: the order parameter can be continuous   at the phase transition  for the disordered system whereas it presents a discontinuity for the pure system  (see e.g. the pioneering work of Imri and Ma \cite{cf:IM}). An instance for which this phenomenon is rigorously proved is the magnetization transition  of the two dimensional random field Ising model at low temperature (see \cite{cf:AW}).

\medskip

This phenomenon has been particularly studied for the polymer pinning  on a defect line (introduced by Fisher in \cite{cf:Fisher}).
Whereas the model can be defined for a renewal with a distribution tail which is heavier than exponential (see \eqref{rules}), the case of power-law tail has focused most of the attention, due to its physical interpretation and its rich mathematical structure. The interested reader can refer to \cite{cf:GB, cf:GB2, cf:denH} for reviews on the subject. The smoothing of the free-energy curve for the pinning model with power-law tails was proved in \cite{cf:GT05} (with some restriction on the law of the disorder see \cite{cf:CdH} for a recent generalization of the result; see also \cite{cf:BL2, cf:LT} for related models). This confirmed predictions by theoretical physicists \cite{cf:DHV} based on an interpretation of the Harris criterion \cite{cf:Harris}. Some other consequences of the introduction of disorder such as critical point shift were studied in \cite{cf:A06, cf:T08, cf:DGLT07, cf:AZ08, cf:GLT08, cf:GLT09, cf:AZ10}. 

\medskip

The present paper aims to study how this phenomenology transposes for renewals with a much lighter tail, stretched exponential ones. Whereas this issue does not seem to be discussed much in the literature, it is clear from a mathematical point of view that the type of argument used in \cite{cf:GT05} do not extend to that case (see Section \ref{smoothpol} for a more detailed discussion). This hints to the fact that when renewal tails gets lighter, Harris predictions on disorder relevance might not apply (or at least not in a straightforward manner). We show that this is indeed the case and provide a necessary and sufficient condition on the return exponent for smoothing of the free-energy curve to hold.

\medskip

Let us notice finally notice that renewals with stretched exponential tails have recently been the object of a study by Torri  \cite{cf:T14} with a different perspective: he focuses on the issue of the scaling limit of the process when the environment is heavy tailed.

\subsection{The disordered pinning model}

Let us shortly introduce the model: set $\tau:=(\tau_0,\tau_1,\ldots)$ to be a renewal process of law
$\bP$, with inter-arrival law $K(\cdot)$, {\sl i.e.}, $\tau_0=0$ and
$\{\tau_i-\tau_{i-1}\}_{i\in\N}$ is a sequence of IID positive integer-valued
random variables. 
Set 
\begin{equation}
K(n):=\bbP[\tau_1=n].
\end{equation}
We assume that 
\begin{equation}\label{rules}
\lim_{n\to \infty} n^{-1} \log K(n)=0.
\end{equation}
Note that with a slight abuse of notation, $\tau$ can also be considered as a subset of $\bbN$ and we will write 
$\{ n\in \tau \}$ for $\{ \exists i, \ \tau_i=n \}$.
 The random potential $\go:=\{\go_1,\go_2,\ldots\}$ is a sequence of
IID centered random variables  which have unit variance and  exponential moments of all order

\begin{equation}\label{defgl}
\gl(\gb):=\log \bbE[e^{\gb\go}]<\infty.
\end{equation}

\medskip

Given $\gb>0$ (the inverse temperature) and $h\in \bbR$,
we define $\bP^{\gb,h,\go}_N$  a measure whose Radon-Nikodym derivative w.r.t $\bP$ is given by 
\begin{equation}\label{density}
\frac{\dd \bP^{\gb,h,\go}_N}{\dd \bP}(\tau):= \frac{1}{Z^{\gb,h,\go}_N}\exp\left(\sum_{n=0}^N(\gb \go_n+h)\delta_n\right)\delta_N
\end{equation}
where $\delta_n=\ind_{\{n\in \tau\}}$ and $Z^{\gb,h,\go}_N$ is the renormalizing constant which makes $\bP^{\gb,h,\go}_N$ a probability law:
\begin{equation}
  \label{eq:Znh}
  Z^{\gb,h,\go}_{N}:=\bE\left[e^{\sum_{n=1}^N(\beta\go_n+h)\delta_n}\delta_N\right].
\end{equation}

\begin{remark} \label{rembc} In the definition \eqref{density} of $\bP^{\gb,h,\go}_N$, the $\delta_N$ corresponds to constraining the end point to be pinned.
This conditioning is present for technical reasons and makes some computations easier but is not essential.
\end{remark}

By ergodic  super-additivity, (see \cite[Chap. 4]{cf:GB}), the limit
\begin{equation}
  \label{eq:F_nh}
  \tf(\beta,h):=\lim_{N\to\infty}\frac1N\log Z^{\gb,h,\go}_{N}
\end{equation}
exists and is non-random. It is non-negative because of assumption \eqref{rules} and convex in $h$ as a limit of convex functions.
The expectation also converges to the same limit
\begin{equation}\label{freen}
  \tf(\beta,h)=\lim_{N\to\infty}\frac1N\bbE\log Z^{\gb,h,\go}_{N}.
\end{equation}
The function $\tf$ is called the free-energy  (or sometimes pressure) of the system. Its derivative in $h$ gives the asymptotic contact fraction of the renewal process, i.e.\ the mean number of contact per unit length,
\begin{equation}\label{gro}
  \partial_h \tf(\beta,h)=\lim_{N\to \infty} \frac 1 N \bE^{\gb,h,\go}\left[\sum_{n=1}^N \delta_n\right].
\end{equation}
The above convergence holds by convexity as soon as $\partial_h \tf(\beta,h)$ is defined (i.e.\ everywhere except eventually at a countable number of points).
If \eqref{rules} holds, the system undergoes a phase transition from a de-pinned state
($\tf(\gb,h)\equiv 0$) to a pinned one ($\tf(\gb,h)>0$ and $ \partial_h \tf(\beta,h)>0$) when $h$ varies.

\medskip

We define $h_c(\gb)$, the critical point at which this transition occurs 
\begin{equation}
h_c(\gb):= \min\left\{ h \ | \ \tf(\gb,h)>0\right\}. 
\end{equation}
As the underlying renewal process $\tau$ is recurrent, we have $h_c(0)=0$.
From \cite[Theorem 2.1]{cf:Alea}, the free energy is infinitely differentiable in $h$ on $(h_c(\gb),\infty)$ (so that \eqref{gro} holds everywhere except maybe at the critical point).
The phase transition  for the pure system, that is, for $\gb=0$, is very well understood. The pure model is said to be exactly solvable 
and there is a closed expression for $\tf(0,h)$ in terms of the renewal function $K$ (see \cite{cf:Fisher}).

\subsection{Disorder relevance and Harris criterion for power-law renewals}

The disordered system ($\gb>0$) is much more complicated to analyze and has given rise to a rich literature, most of which devoted to the case where when $n\to \infty$
\begin{equation}\label{poly}
K(n)=c_K n^{-(1+\alpha)}(1+o(1))
\end{equation}
for some $\alpha>0$. For the pure model, the free-energy vanishes like a power of $h$ at $0+$ (see \cite[Theorem 2.1]{cf:GB}).
\begin{equation}
  \tf(0,h)= c'_K h^{\max(1,\alpha^{-1})}(1+o(1)),
 \end{equation}
for $\alpha\ne 1$ (a logarithmic correction is present in the case $\alpha=1$).
The main question for the study of disordered pinning model is how this property of the phase transition is affected by the introduction of disorder.
For $\gb>0$, does there exist $\nu$ such that at the vicinity of  $h_c(\gb)_+$ 
\begin{equation}\label{nu}
  \tf(\gb,h)\approx (h-h_c(\gb))^{\nu}?
 \end{equation}
If this holds, is $\nu$ equal $\max(1,\alpha^{-1})$, like for the pure system?
A first partial answer to that question was given by Giacomin and Toninelli \cite{cf:GT05} (or in \cite{cf:CdH} with more generality)  where it was shown that 
\begin{equation}
  \tf(\gb,h)\le C\left( \frac{h-h_c(\gb)}{\gb}\right)^{2},
 \end{equation}
meaning that the quenched critical exponent for the free-energy $\nu$, if its exists, satisfies $\nu \ge 2$.
In particular it cannot be equal to the one of the pure system when $\alpha>1/2$.

\medskip

One the other hand, for small $\gb$ and $\alpha<1/2$, it was shown by Alexander \cite{cf:A06} (see \cite{cf:T08, cf:Lmart} for alternative proofs) that $h_c(\gb)=-\gl(\gb)$ (recall \eqref{defgl})  and that when $u\to 0+$
\begin{equation}
  \tf(\gb,u-\gl(\gb))=  \tf(0,u)(1+o(1))
\end{equation}
meaning that $\nu$ exists and is equal to $\max(1,\alpha^{-1})$ as for the pure model.

\medskip

Another aspect of the relevance of disorder is the shift of the \textsl{quenched} critical point with respect to the \text{annealed} one.
The annealed critical point is the one corresponding to the phase transition of the annealed partition function obtained by averaging over the environment
\begin{equation}\label{annealedbound}
h_c^a(\gb):=\inf \left\{ h \ | \ \lim_{N\to \infty} \frac{1}{N} \log  \bbE\left[ Z^{\gb,h,\go}_{N}\right]>0\right\}=-\gl(\gb).
\end{equation}
It follows from Jensen's inequality that 
\begin{equation}\label{quenchann}
h_c(\gb)\ge h^a_c(\gb).
\end{equation}
The question of whether the above inequality is strict
was investigated in 
 \cite{cf:DGLT07, cf:AZ08, cf:GLT08, cf:GLT09} yielding the conclusion that $h_c(\gb)>-\gl(\gb)$ for every $\gb>0$ and $\alpha\ge 1/2$.

\medskip

These results were predicted in the Physics literature \cite{cf:DHV, cf:FLN}, based on an interpretation of the Harris criterion \cite{cf:Harris}: when the specific-heat exponent of the pure system (for the pinning model, this exponent is equal to $2-\max(1,\alpha^{-1})$) is positive, then disorder affects the critical properties of the system and is said to be relevant, whereas when  the specific-heat exponent is negative disorder is irrelevant for small values of $\gb$.

\medskip

Relevant disorder affects both the location of the critical point which is shifted with respect to the annealed bound \eqref{quenchann}
 \cite{cf:DGLT07, cf:AZ08,cf:GLT08, cf:GLT09}, and the critical exponent of the free-energy \cite{cf:GT05, cf:CdH}.
Note that the value of $\nu$ (and even its existence)  when disorder is relevant is an open question even among physicists; let us mention the recent work 
\cite{cf:DR} where heuristics in favor $\nu=\infty$ (infinitely derivable free-energy at the critical point) are  given for a toy-model.

%
%


In this paper, we choose to look at renewal processes whose tails are stretched exponentials, i.e we assume that there exists $\zeta\in(0,1)$ such that
\begin{equation} 
K(n)\approx \exp(-n^{\zeta}),
\end{equation}
in some sense. 
As the increments of $\tau$ have finite mean, the transition of the pure model is of first order, meaning that $\tf(0,h)$ is not derivable at $h_c(0)=0$ positive recurrent. More precisely,
from \cite[Th. 2.1]{cf:GB} one has
\begin{equation}
  \label{eq:annF}
  \tf(0,h)\stackrel{h\searrow0}\sim \frac{h}{\bE[\tau_0]}.
  \end{equation}
 as for the case $\alpha>1$ in \eqref{poly}. Hence a standard interpretation of the Harris criterion would tell us that disorder should be relevant for every $\gb$.
This is partially true in the sense that this conclusion is right if one considers only the question of the critical point shift.
The method developed  in \cite{cf:DGLT07} can be adapted almost in a straightforward manner to show that

\begin{proposition}\label{mook}
When $K(n)$ has stretched-exponential tails, then for all $\gb>0$,
\begin{equation}
h_c(\beta)>-\gl(\gb).
\end{equation}
\end{proposition}

The more challenging question is the one about the order of the phase transition. Indeed the smoothing inequality proved in \cite{cf:GT05} strongly relies of the fact that
$K(\cdot)$ has a power-law tail. 

\medskip

We are in fact able to find a necessary and sufficient condition on $\zeta$ for a smoothing inequality to hold: we prove that when $\zeta>1/2$, the transition remains of first order for the disordered system, while for $\zeta\le 1/2$ the transition is rounded.
We also give upper and lower bounds, which do not coincide, on the exponent $\nu$, informally defined in \eqref{nu}, when rounding occurs, in particular we show that for any value of $\zeta\in(0,1)$, the disordered phase transition remains of finite order.

\section{Presentation of the results} \label{sec:nHmodel}

\subsection{Results}

We assume here and in what follows that there exist a constant $c_K$ and $\zeta\in (0,1)$ which is such that
\begin{equation}
  \label{eq:K}
K(n)= c_K(1+o(1))\exp(-n^{\zeta}).
\end{equation}
The law $K(n)$ as well as the law of $\go$ are considered to be fixed, and constants that are mentioned throughout the proof can depend on both.
Unless it is specified, they will not depend on $\gb$ and $h$.


For our first result, we need to assume that the law of our product environment satisfies a concentration inequality.
We say that $F: \bbR^N\to \bbR$ is Lipschitz if for some $k>0$ if 

\begin{equation}
\| F \|_{\lip}=\sup_{x\ne y \in \bbR^N} \frac{|F(x)-F(y)|}{|x-y|}<\infty
\end{equation}
where $|x-y|= \sqrt{ \sum (x_i-y_i)^2}$  is the Euclidean norm.

\begin{assumption}\label{logsob}
There exist constants $C_1$ and $C_2$ such that for any $N$ and for any Lipschitz convex function $F$ on $\bbR^N$, one has 
\begin{equation}
\bbP \left(
| F(\go_1,\dots,\go_N)-\bbE \left[ F(\go_1,\dots,\go_N)\right]| \ge u  \right) \le C_1e^{-\frac{u^2}{C_2 \| F \|_{\lip}^2}}
\end{equation}
\end{assumption}

A crucial point here is that inequality does not depend of the dimension $N$. This is the reason why we use concentration for the Euclidean norm rather than for the $L_1$ norm.

\begin{remark}
The concentration assumption is not very restrictive, it holds for bounded  $\go$ (see \cite[Chapter 4]{cf:Ledoux}),
or when $\go$ satisfies a $\log$-Sobolev inequality (see \cite[Chapter 5]{cf:Ledoux} in this case there is no convexity required). This second case includes in particular the case of Gaussian variables and many other classic laws.
\end{remark}

Our first result states that the transition is of first order for the system for $\zeta>1/2$ (no smoothing holds). Here and in what follows $x_+:= \max(x,0)$ denotes the positive part of $x\in \bbR$.

\begin{theorem}\label{mainres}
Assume that Assumption \ref{logsob} holds.
\begin{itemize}
\item[(i)] For $\zeta>1/2$, there exists a constant $c$ such that for all $\gb$ and $h$,

\begin{equation}\label{ground}
\tf(\gb,h)\ge c\max(1,\gb^{-2})  (h-h_c(\gb))_+.
\end{equation}

\item[(ii)]
For $\zeta\le 1/2,$ there exists a constant $c$ and $u_0(\gb)$ such that for all $u\in(0,u_0)$
\begin{equation}\label{ground2}
\tf(\gb,h_c(\gb)+u )\ge c\left(\frac{u}{\gb^2 | \log u|}\right)^{\frac{1-\zeta}{\zeta}}.
\end{equation}
\end{itemize}
 \end{theorem}

\medskip

Our second result shows that smoothing holds for $\zeta< 1/2$.
For this result we need to assume that the environment is Gaussian. 
The assumption could be partially relaxed but the exposition of the Gaussian case is much easier.
Let us mention that the recent work \cite{cf:CdH} gives hopes to extend the proof to general $\go$.

\begin{theorem}\label{mainres2}
Let us assume that the environment is Gaussian. Then for all $\zeta<1/2$ there exists a constant $c$ (which depends on $K$) such that in a neighborhood of $h_c(\gb)$
\begin{equation}
\tf(\gb,h)\le c \left(\frac{h-h_c(\gb)}{\gb}\right)^{2(1-\zeta)}_+.
\end{equation}
\end{theorem}

Finally with an extra assumption on $K(\cdot)$ we are able to state that the transition is smooth also when $\zeta=1/2$.
We say that $K(n)$ is $\log$ convex if $\log K$ can be extended to a convex function on $\bbR_+$;  or equivalently if  one has 
\begin{equation}\label{logconvex}
\forall n,l \in \bbN, \ n>l>1 \Rightarrow K(n+1)K(l-1)\ge K(n)K(l).
\end{equation}
This assumption is necessary to prove positive correlation, or the FKG inequality (see \cite{cf:FKG}) for the disordered renewal.

\begin{theorem}\label{mainres3}
Assume that $\log K(n)$ is a convex function of $n$. Then for $\zeta=1/2$ one has 
\begin{equation}
\tf(\gb,h)= o\left((h-h_c(\gb))_+\right).
\end{equation}
\end{theorem}

\begin{remark}
The $\log$-convex assumption is not that restrictive and is rather natural as assumption \eqref{rules} already implies that the derivative of $K$ tends to zero.
A particular instance of $\log$-convex $K$ is the case where $\tau$ is the set of return times to zero  of a one dimensional nearest-neighbor random walk on $\bbZ$. This is related to $\log$-convexity  of  the  sequence of Catalan numbers (see \cite{cf:BF} for a paper on the subject).
\end{remark}

\subsection{The smoothing for polynomial tails}\label{smoothpol}

Let us explain briefly in this section why the strategy from \cite{cf:GT05} fails to give any results in the case of stretched exponential renewals (for more details the reader should refer to the original article). For simplicity we assume here that the environment is Gaussian and that $\gb=1$.

\medskip

The main idea in \cite{cf:GT05} is the following. Let $h=h_c(\gb)+u$ be fixed, and $N$ be chosen very large.
We look at a system at the critical point $h_c(\gb)$  (for which the free energy is zero): 
in a typical segment of length $N$ the empirical mean of $\go$ should be of order $0$ due to the law of large number ; however, with probability of order $\exp(- N u^2/2 )$ the empirical mean is larger than $u$. In that case, the system does not locally look critical and the partition function corresponding to the segment should be of order $e^{N \tf(\gb,h_c(\gb)+u)}$, if $N$ is chosen sufficiently large to avoid finite size effects.

\medskip

The distance between these segments of length $N$ which give an unusually "good" contribution to the partition function should be typically huge, that is, of order $\exp(u^2 N/2 )$, and thus the cost for making a huge jump between two consecutive good segments to avoid bad environment should be of order 
$K(\exp(u^2 N/2))$. As the free-energy at criticality is zero, the strategy consisting in visiting all the "good" segments and avoiding all the bad ones should not give an exponentially large contribution to the partition function, hence the cost of making the large jump should completely compensate for the energy reward one gets when visiting a good segment. 
For this reason one must have for sufficiently large $N$

\begin{equation}\label{kaboom}
K(\exp(u^2 N/2))e^{N \tf(\gb,h_c(\gb)+u)} < 1.
\end{equation}

In the case where $K$ has a power-law tail, this immediately yields a quadratic bound on the free energy. However, when $K$
has a lighter tail, \eqref{kaboom} fails to give any interesting information, as $K(\exp(u^2 N/2))$ decays super-exponentially.

\medskip

Some elements of this strategy can somehow be recycled (this is what is done in Section \ref{seclwb}) if one has some information about the behavior of finite volume systems (see Lemma \ref{finitevol}) . However, as will be seen, this fails to give a quadratic bound on the free-energy.

\subsection{Comparison with the case of renewals with exponential and sub-exponential tails}

An other instance of pinning model with absence of smoothing has been exhibited in  \cite{cf:Aivy}: disordered pinning of transient renewals with exponential tails 
(the case $K(n)=O(\exp(-n b))$ for some $b>0$).
However, let us mention that this is case quite special since when the tail of the renewal is exponential Remark \ref{rembc} is not valid anymore. On the contrary, the behavior of the system crucially depends on the contraint one imposes at the end point:
\begin{itemize}
\item The free-energy $\tf(\gb,h)$ defined by \eqref{freen}, which corresponds to a system constrained to be pinned,  is negative for small values of $h$.
\item The free energy of the system with no constraints is obtained by considering the best of two strategies: either the walk will avoid the wall completely or it will try to pin the end point. The reward for this is equal to $\max(0,\tf(\gb,h))$, which is easily shown to have a first order transition in $h$. 
\end{itemize}
Here the mechanism which triggers a first order phase transition is completely different: one has to perform an analysis of the  local fluctuations of the environment to see whether or not the benefit of a good rare region is sufficient to compensate the cost of a large jump reaching it. An upper bound on the fluctuations is obtained via concentration.
To obtain a lower-bound, we choose to restrict to the Gaussian model for simplicity, but similar ideas could in principle be implemented by the use of tilting (like in \cite{cf:CdH}).

\section{Preliminaries}

\subsection{Notation}

The dependence in $\gb$ and $h$ will frequently be omitted to lighten the notation.
When $A$ is an event for $\tau$ we set 
\begin{equation}
  \label{eq:event}
  Z^\go_{N}(A):=\bE\left[e^{\sum_{n=1}^N(\beta\go_n+h)\delta_n}\delta_N\ind_A\right].
\end{equation}
For $k\in \bbN$, the shift operator $\theta^k$ acting on the sequence $\go$ is defined by 
\begin{equation} 
\theta^k \go_n:= \go_{n+k}.
\end{equation}
For any couple of integers $a\le b$ we set 
\begin{equation}\label{znab}
 Z^{\go}_{[a,b]}=e^{(\gb \go_a+h)\ind_{a> 0}}  Z^{\theta^a \go}_{b-a}.
\end{equation}
to be the partition function associated to the segment $[a,b]$ (with the convention that $Z^\go_0=1$). Note that the environment at the starting point of the interval $a$ is taken into account only for $a>0$ (for technical reasons).

\medskip
For $\gep>0$ we define
\begin{equation}
\label{abepsilon}
\begin{split}
\cA^\gep&:=\{ \tau \ | \ \#( \tau\cap (0,N])\  |  \le \gep N, N\in \tau \},\\
\cB^\gep&:=\{ \tau \ | \ \#( \tau\cap (0,N])\  |  > \gep N, N\in \tau \},
\end{split}
\end{equation}
the set of renewals whose contact fraction is smaller, resp. larger, than $\gep$.

\subsection{Finite volume bounds for the free energy}
The following result allows to estimate the free-energy only knowing the value of $\frac{1}{N}\bbE\left[\log Z^\go_N\right]$, for a given $N$.

\begin{lemma}\label{finitevol}
There exists a constant $c$ such that for every $N$, $\gb$ and $h$, 
\begin{equation}\begin{split}
\frac{1}{N}\bbE\left[\log Z^\go_N\right]& \le \tf(\gb,h),\\
\frac{1}{N}\bbE\left[\log Z^\go_N\right]& \ge \tf(\gb,h)- \frac{N^{\zeta-1}}{1-2^{\zeta-1}}-\frac{2(\gl(\gb)+h)_++c}{N}
\end{split}
\end{equation}

\end{lemma}
\begin{proof}
The first inequality is a consequence of following super-multiplicativity  property 
\begin{equation}
Z^{\go}_{N+M}\ge Z^{\go}_{N}\times Z^{\theta^N \go}_{M}
\end{equation}
(see e.g.\ the proof of \cite[Proposition 4.2]{cf:GB}).
For the second one the proof is similar to \cite[Proposition 2.7]{cf:Alea}, one has 

\begin{equation}\label{anelka}
Z^\go_{2N}= \bE\left[e^{\sum_{n=1}^N(\beta\go_n+h)\delta_n}\delta_{N}\delta_{2N}\right]
+\bE\left[e^{\sum_{n=1}^N(\beta\go_n+h)\delta_n}(1-\delta_{N})\delta_{2N}\right].
\end{equation}
The first term is equal to 
$Z^\go_{N}Z^{\theta^N \go}_{N}.$
For the second term, by comparing the weight of each $\tau$ to the one of $\tau\cup\{N\}$ one obtains 
\begin{multline}\label{patagony}
\bE\left[e^{\sum_{n=1}^N(\beta\go_n+h)\delta_n}(1-\delta_{N})\delta_{2N}\right]\\
\le  Z^\go_{N}Z^{\theta^N \go}_{N}e^{-\gb \go_N-h}\max_{0\le a<N<b\le 2N} \frac{K(b-a)}{K(N-a)K(b-N)}\\
\le C e^{-\gb \go_N-h} Z^\go_{N}Z^{\theta^N \go}_{N}\exp(N^{\zeta}),
\end{multline}
for some constant $C>1$.
The last line of \eqref{patagony} is obtained by observing that for any choice of $0\le a<N<b\le 2N$
$$(N-a)^\zeta	+(b-N)^\zeta-(b-a)^\zeta\le N^{\zeta}.$$

Hence taking the $\log$ and expectation in \eqref{anelka} one has 

\begin{equation}\begin{split}
\frac{1}{2N}\bbE\left[ \log Z^\go_{2N}\right]&\le \frac{1}{N}\bbE\left[\log Z^\go_{N}\right]+ \frac{1}{N} \bbE\left[\log \left( 1+Ce^{-\gb \go_N-h}  \exp(N^{\zeta})\right)\right]
\\
&\le \frac{1}{N}\bbE\left[ \log Z^\go_{N}\right]+\frac{1}{N} \log \left( 1+e^{\gl(-\gb)-h}C \exp(N^{\zeta})\right)\\
&\le \frac{1}{N}\bbE\left[ \log Z^\go_{N}\right]+N^{\zeta-1}+\frac{1}{N} \log \left( 1+ C e^{\gl(-\gb)-h}\right),\\
&\le \frac{1}{N}\bbE\left[ \log Z^\go_{N}\right]+N^{\zeta-1}+\frac{1}{N} \left[ \log(2C) + (\gl(-\gb)-h)_+ \right],\\
\end{split}\end{equation}
where the first inequality is simply Jensen's inequality.
Then we iterate the inequality and obtain 
\begin{equation}
\tf(\gb,h)\le \frac{1}{N}\bbE\left[ \log Z^\go_{N}\right]+\frac{N^{\zeta-1}}{1-2^{\zeta-1}}+\frac{2}{N} \left[ \log(2C) + (\gl(-\gb)-h)_+ \right].
\end{equation}
\end{proof}

\subsection{The FKG inequality for $\log$-convex renewals}

For the proof of Theorem \ref{mainres3} (and only then), we need to use the fact that the presence of renewal point are positively correlated.
This is where we need the assumption on the $\log$ convexity of the function $K$. 

\medskip

In this subsection $\tau$ denotes a subset of $\{1,\dots,N\}$ which contains $N$,
and with some abuse of notation  $\bP^{\gb,h,\go}_N$ is considered to be a law on $\cP(\{1,\dots,N\})$ (the set of subsets of $\{1,\dots,N\}$).

\medskip

Now let us introduce some definitions.
A function $f: \cP(\{1,\dots,N\})\to \bbR $ is said to be increasing if 
\begin{equation}
\forall \tau, \tau' \in \cP(\{1,\dots,N\})\quad \tau\subset \tau' \Rightarrow  f(\tau)\le f(\tau').
\end{equation}

Note that the following result was proved in \cite{cf:BW} for renewal processes in continuous time. Our proof is essentially similar and is based on the use of the 
celebrated FKG criterion from \cite{cf:FKG} but we choose to include it for the sake of completeness. 

\begin{proposition}
Assume that the function $K$ is $\log$-convex.
Then for all $\gb, \go, h$ and $N$, the $\bP^{\gb,h,\go}_N$ satisfies the FKG inequality.  Namely, for all increasing functions $f$ and $g$
\begin{equation}
\bE^{\gb,h,\go}_N[f(\tau) g(\tau)]\ge \bE^{\gb,h,\go}_N[f(\tau)] \bE^{\gb,h,\go}_N[g(\tau)].
\end{equation}
\end{proposition}

\begin{proof}
From \cite[Proposition 1]{cf:FKG}, it is sufficient to check that for any $\tau$ and $\tau'$ one has 

\begin{equation}\label{ooo}
\bP_N^{\gb,h,\go}(\tau\cup \tau') \bP_N^{\gb,h,\go}(\tau\cap \tau') \ge \bP_N^{\gb,h,\go}(\tau) \bP_N^{\gb,h,\go}(\tau').
\end{equation}
For $\sigma\subset \{0,\dots,N\}$ whose elements are $\sigma_0=0<\sigma_1<\dots<\sigma_m=N$,  we set 
$$K(\sigma)=\prod_{i=1}^m K(\sigma_i-\sigma_{i-1}).$$
The reader can check that after simplification \eqref{ooo} is equivalent to 

\begin{equation}\label{toto}
K(\tau\cup \tau')K(\tau\cap \tau') \ge K(\tau) K(\tau').
\end{equation}
This inequality is obviously true when $\tau'\subset \tau$. Then we proceed by induction and it is sufficient to
 check  that if $a\notin \tau\cup \tau' $ and \eqref{toto} holds for $\tau$ and $\tau'$, then it holds for 
$\tau$ and $\tau'\cup\{a\}$. To this purpose we only need to check that for any $\tau$, $\tau'$ and $a\notin \tau\cup \tau' $ we have 
\begin{equation}\label{hopla}
\frac{K(\tau\cup \tau' \cup\{a\})}{K(\tau\cup \tau')}\ge \frac{K(\tau' \cup\{a\})}{K(\tau')}.
\end{equation}
Let us set
\begin{equation}
\begin{split} 
\alpha_1:=\inf\{ x < a \ | \ x\in \tau\cup \tau' \}, \quad & \beta_1:=\inf\{ x > a \ | \ x\in \tau\cup \tau' \},\\
\alpha_2:=\inf\{ x < a \ | \ x\in \tau' \}, \quad & \beta_2:=\inf\{ x > a \ | \ x\in  \tau' \}.
\end{split}
\end{equation}
We remark that $$\alpha_2\le \alpha_1<a<\beta_1\le \beta_2.$$
The inequality \eqref{hopla} is equivalent to 
\begin{equation}\label{totti}
 \frac{K(\beta_1-a)K(a-\alpha_1)}{K(\beta_1-\alpha_1)}\ge \frac{K(\beta_2-a)K(a-\alpha_2)}{K(\beta_2-\alpha_2)}.
\end{equation}
By convexity of $\log K$ the function  

\begin{equation}
(\alpha,\gb)\mapsto   \frac{K(\beta-a)K(a-\alpha)}{K(\beta-\alpha)}
\end{equation} is non-increasing in $\gb$ and non-decreasing in $\alpha$ on the set $\{ (\alpha,\beta)  \ | \ \alpha<a<\beta \}$ . Thus \eqref{totti} holds. 
\end{proof}

\section{Proof of Theorem \ref{mainres}}

\subsection{Decomposition of the proof}

The key point consists in proving the following upper bound on $Z_N^\go(\cA^\gep)$ (recall \eqref{abepsilon} and \eqref{eq:Znh}).

\begin{proposition}\label{keypoint}
There exist positive constants $\gep_0$ and $C$ such that for all $\gep\le \gep_0$
we have for all $h\le 1$ and $\gb>0$, , almost surely for all $N$ sufficiently large,
\begin{equation}\label{orange}
\frac{1}{N} \log Z^\go_N(\cA^\gep)\le \frac{1}{2}\tf(h,\gb) + \max_{l\ge \gep^{-1}} \left( C\gb \sqrt{\frac{\gep \log l}{l}}-\frac{1}{4}l^{\zeta-1}  \right).
\end{equation}
\end{proposition}

The restriction $h\le 1$ is chosen for convenience but does not convey any particular significance ($h<c$ for some $c>0$ would be just as good).
The proof of this statement is postponed to Section \ref{quatpoind}.

Now, we observe that if $\gep$ is chosen to be larger than the asymptotic contact fraction $\partial_h \tf(\gb,h)$, the l.h.s.\ of \eqref{orange} converges to the the free-energy.
\begin{lemma}\label{fractio}
For every $h>h_c(\gb)$ when $\gep>\partial_h\tf(\gb,h)$ one has.

\begin{equation}\label{limimi}
 \liminf_{N\to \infty} \bP^{\gb,h,\go}_N\left[\cA^\gep\right]>0.
\end{equation}
As a consequence 
\begin{equation}
\limsup_{N\to \infty} \frac{1}{N}  \log Z^\go_N(\cA^\gep)=\tf(\gb,h).
\end{equation}
\end{lemma}
\begin{remark}
Without much more efforts, one can even  prove in fact that the limit in \eqref{limimi} is equal to one, but this is not necessary for our purpose.
\end{remark}

The idea to prove Theorem \ref{mainres} is to use \eqref{orange} where $\gep$ is replaced by $2\partial_h \tf(\gb,h)$
and $\frac{1}{N} \log Z_N^\go(\cA^\gep)$ is replaced by its limit: $\tf(\gb,h)$. This gives a differential inequality in $h$ which once integrated gives the claimed bounds on the free energy. Details follow at the end of the Section.

\begin{proof}[Proof of Lemma \ref{fractio}]
For simplicity (and with no loss in generality) assume that $\gep=(1+\delta) \partial_h\tf(\gb,h)$ for some $\delta<1$.
By \eqref{gro}, for $N$ sufficiently large 
\begin{equation}
\frac{1}{N}\bE^{\gb,h,\go}_N\left[\sum_{n=1}^N \delta_n\right]\le (1+\delta)(1-\delta/2) \partial_h \tf(\gb,h)=(1-\delta/2) \gep.
\end{equation}
As 
\begin{equation}
\frac{1}{N}\bE^{\gb,h,\go}_N\left[\sum_{n=1}^N \delta_n\right]\ge \gep \bP^{\gb,h,\go}_N \left[\cB^\gep\right],
\end{equation}
this implies 
\begin{equation}
 \bP^{\gb,h,\go}_N\left[\cB^\gep\right]\le 1-\delta/2.
\end{equation}
\end{proof}

\begin{proof}[Proof of Theorem \ref{mainres}]
Let $\gep_0$ be such that  Proposition \ref{keypoint} holds (in the case $\zeta>1/2$, we will also require $\gep_0$ to satisfy another condition).
If $h\le 1$ is such that 
\begin{equation}\label{defeph} 
\gep_h:=2\partial_h\tf(\gb,h)\le \gep_0,
\end{equation} 
we have
\begin{equation}\label{cramook}
\tf(\gb,h) \le  \max_{l\ge \gep_h^{-1}}\left( 2C\gb\sqrt{\frac{\gep_h \log l}{l}}- \frac{1}{2}l^{\zeta-1}  \right).
\end{equation}
Let us start with the case $\zeta>1/2$. By contradiction, let us  assume that,
\begin{equation}\label{mangepasdpin}
\lim_{h\to h_c(\gb)+} \partial_h \tf(\gb,h)< \frac{1}{2}\gep_0\min(1,\gb^{-2}).
\end{equation}
From a standard convexity argument (see \cite[Proposition 5.1]{cf:GB}) one has $h_c(\gb)\le 0$.  Thus we can find 
$h\in (h_c(\gb), 1]$  such that 
$$\gep_h\le \gep_0\min(1,\gb^{-2}).$$
Then for this value of $h$ the right-hand side of \eqref{cramook} is smaller than 

$$ \sup_{l\ge \gep_0^{-1}}\left( 2C\sqrt{\frac{\gep_0 \log l}{l}}- \frac{1}{2}l^{\zeta-1}  \right)$$
which is equal to zero if $\gep_0$ has been chosen sufficiently small. Hence we obtain a contradiction as $\tf(\gb,h)>0$.

\medskip

Let us move to the case $\zeta\le1/2$. We can assume that 
\begin{equation}\label{mangepasdpin2}
\lim_{h\to h_c(\gb)+} \partial_h \tf(\gb,h)=0 
\end{equation}
as if not, there is nothing to prove.

\medskip

For $h$ sufficiently close to the critical point we hence have $\gep_h\le \gep_0$ (and $h\le 1$)
and hence 
Equation \eqref{cramook} holds. 

Computing the maximum in the r.h.s.\ of \eqref{cramook} we obtain

\begin{equation}\label{ineqps}
\tf(\gb,h)\le \begin{cases} C\left( \gb^2 \gep_h |\log \gep_h |\right)^{\frac{1-\zeta}{1-2\zeta}} &\text{ for } \zeta<1/2 \\
\exp(- c(\gb^2 \gep_h)^{-1}) &\text{ for } \zeta= 1/2.
\end{cases}
\end{equation}

When $\zeta<1/2$, we note that the inverse of the function

$$f: x \to C \left( \gb^2 x | \log x |\right) ^{\frac{1-\zeta}{1-2\zeta}} $$
(which is increasing near zero)
satisfies 

\begin{equation}
f^{-1}(y)\stackrel{y\to 0}  {\sim} C ' \gb^{-2} y^{\frac{1-2\zeta}{1-\zeta}}   | \log y |^{-1}.
\end{equation}

Hence composing the inequality \eqref{ineqps} with $f^{-1}$ and replacing $\gep_h$ by its value  \eqref{defeph}, we obtain that in the vicinity of $h_c(\gb)_+$ we have

\begin{equation}
\partial_h \tf \ge c \gb^{-2} \tf^{\frac{1-2\zeta}{1-\zeta}} | \log \tf |^{-1},
\end{equation}
and hence

\begin{equation}\label{tobeintegrated}
\tf^{\frac{2\zeta-1}{1-\zeta}} |\log \tf| \partial_h  \tf \ge c \gb^{-2}.
\end{equation}
It is easy to check that this inequality is valid also in the case $\zeta=1/2$.
Now at the cost of modifying the constant $c$ (and the neighborhood of $h_c(\gb)$ if necessary)
the inequality implies

\begin{equation}\label{tobeintegrated2}
\partial_h \left[ \tf^{\frac{\zeta}{1-\zeta}}   |\log \tf|  \right] \ge c \gb^{-2},
\end{equation}
which once integrated implies 

\begin{equation}
 \tf(\gb,h_c(\gb)+u)^{\frac{\zeta}{1-\zeta}}   |\log \tf(\gb,h_c(\gb+u)|\ge c\gb^{-2}u.
\end{equation}

Composing the inequality with the inverse of the function $x\mapsto  x^{\frac{\zeta}{1-\zeta}}   |\log x|$ near zero gives the desired result.

Integrating the above inequality between $h_c(\gb)$ and $h$ yields the result.

\end{proof}

\subsection{Proof of Proposition \ref{keypoint}}\label{quatpoind}

A key tool in the proof is the following concentration inequality.

\begin{lemma}\label{concentration}
When Assumption \ref{logsob} holds then
for any event $A\subset \cA^\gep$

\begin{equation}
\bbP\left[\log Z^{\go}_N(A)-\bbE[\log Z^{\go}_N(A)]\ge t\right]\le C_1 \exp\left(-\frac{t^2}{C_2\gb^2 N\gep}\right).
\end{equation}
\end{lemma}

\begin{proof}

For any pair of environment $\go$ and $\go'$ one has

\begin{equation}
\left | \log \frac{ Z^{\go}_N(A)}{Z^{\go'}_N(A)} \right |\le \gb \max_{\big\{ \tau \subset [0,N] \ | \ |\tau \cap [0,N]|\le \gep N  \big\}} \sum_{x\in \tau}
|\go_x-\go'_x|\le \gb \sqrt{\gep N} \sqrt{\sum_{x=1}^N \go^2_x}. 
\end{equation}
Hence the Lipshitz norm of
$$\go\mapsto  \log  Z^{\go}_N(A)$$ is smaller than $\gb\sqrt{\gep N}$. It is also a convex function, thus the results follows from Assumption \ref{logsob}.
\end{proof}

Given $\tau \in \cA^\gep$, we define $\cL(\tau)$  the set of indices corresponding to the renewal jumps of length larger than $(2\gep)^{-1}$, and we denote by $L(\tau)$ the cardinal of $\cL(\tau)$.
\begin{equation}
\begin{split}
\cL(\tau)&:=\{ n \ | \ \tau_{n}\le N, (\tau_{n}-\tau_{n-1})\ge (2\gep)^{-1} \},\\
L(\tau)&:=\#\cL(\tau).  
\end{split}\end{equation}
We also set $l(\tau)=N/L(\tau)$.
Due to the definition of $\cA^\gep$ one has 
\begin{equation}
\forall \tau \in \cA^\gep, \quad \sum_{n\in \cL(\tau)}  (\tau_{n}-\tau_{n-1})\ge \frac{N}{2}.
\end{equation}
This means in particular than $l$ is roughly the mean length of $(\tau_{n}-\tau_{n-1})_{n\in \cL(\tau)}$ (up to a factor $2$).
For a fixed $L\in \bbN$, $L\le \gep N$, we set 
\begin{multline}
\cT(L):= \left\{ (\bt',\bt)\in ([0,N]\cap \bbZ)^{2L}  \ \big| \  \forall i\in [1,L], \  t'_i\ge t_{i-1},  \ t_{i}\ge t'_i+(2\gep)^{-1} \right\}\\
\cap \left\{ \sum_{i=1}^L (t_i-t'_i)\ge N/2\right\},
\end{multline}
which is the set of  possible locations for  $(\tau_{n-1},\tau_n)_{n\in \cL(\tau)}$.
For $(\bt',\bt)\in \cT(L)$ we set
\begin{equation} 
A_{(\bt',\bt)}:=\left\{  \tau \ \big| \ \{ (\tau_{n-1},\tau_n) \}_{n\in \cL(\tau)}=  \{ (t'_i ,t_i) \}_{i=1}^L  \ \right\}\cap \cA^{\gep}.
\end{equation}
It is the subset of $\cA^\gep$  for which the jumps of $\tau$ which are longer than $(2\gep)^{-1}$ exactly span the segments 
$(t'_i ,t_i)_{i=1}^L$  (see also Figure \ref{attprim}). 

\begin{figure}[hlt]
 \epsfxsize =11.5 cm
 \begin{center}
 \psfrag{O}{$0$}
  \psfrag{N}{$N$}
  \psfrag{v1}{$t'_1$}
    \psfrag{v2}{$t'_2$}
      \psfrag{t1}{$t_1$}
    \psfrag{t3}{$t_3$}
          \psfrag{t2=v3}{$t_2=t'_3$}
\epsfbox{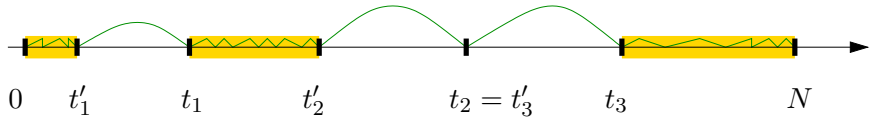}
 \end{center}
 \caption{\label{attprim}  
A schematic representation of a set $(\bt',\bt)\in \cT(3)$, and a renewal in $\tau\in A_{(\bt',\bt)}$ (in green).
The total number of jumps must be smaller than $\gep N$ and in yellow regions, the jumps of $\tau$ must be shorter than $(2\gep)^{-1}$.
As a consequence of these two conditions the total length of the yellow regions is smaller than $N/2$.}
 \end{figure}

We have 
\begin{equation}
Z_N^\go(\cA^\gep)= \sum_{L=1}^{\gep N} \sum_{(\bt',\bt)\in \cT(L)} Z_N^\go (A_{(\bt',\bt)}).
\end{equation}
In particular 

\begin{equation}\label{cramnik}
\log Z_N^\go(\cA^\gep) \le \log N + \max_{L\in\{1,\dots,\gep N\}} \left[ \log (\# \cT(L))+ \max_{(\bt',\bt) \in \cT(L) } \log  Z_N^\go (A_{(\bt',\bt)} )\right].
\end{equation}
The idea is then to use Lemma \ref{concentration} to find a good bound on the l.h.s.
A first easy step is to get an estimate on the cardinal of $\cT(L)$. Recall that here and in what follows $l:=N/L$.

\begin{lemma}\label{one}
There exists a $C$ such that for all $\gep\le 1/4$ for all $N$ sufficiently large
\begin{equation}
\# \cT(L)\le  C \exp( 2 L \log l ). 
\end{equation}
\end{lemma} 
\begin{proof}
The set $\{ t_i\}_{i=1}^L\cup \{ t'_i+1\}_{i=1}^L$ is a subset of 
$\{1,\dots, N\}$ with $2L$ elements. Hence 
\begin{equation}
\# \cT(L)\le \binom{N}{2L}\le C \exp( 2 L \log l ),
\end{equation}
where the last inequality just comes from Stirling formula.
\end{proof}
To use Lemma \ref{concentration} efficiently, we must also know about the expected value of $\log Z^{\go}_N (A_{(\bt',\bt)})$

\begin{lemma}\label{shanti}
For any  $(\bt',\bt) \in \cT(L)$, one has, for $\gep$ sufficiently small (depending only on $K$)

\begin{equation}
\frac{1}{N}\bbE[ \log Z^{\go}_N (A_{(\bt',\bt)})] \le \frac{1}{2}\left( \tf(\gb,h) + l^{-1} -  l^{\zeta-1}\right).
\end{equation}

\end{lemma}

\begin{proof}
One has (recall \eqref{znab})

\begin{equation}
Z^{\go}_N (A_{(\bt',\bt)})\le \left[ \prod_{i=1}^L Z_{[t_{i-1},t'_{i}]}K(t_i-t'_i)\right] Z_{[t_L,N]},
\end{equation}
where we take the convention $t_0=0$.

Hence
\begin{equation} \label{sketletor}
\bbE\left[\log Z^{\go}_N (A_{(\bt',\bt)}\right]\le  \sum_{i=1}^L \bbE\left[\log Z_{[t_{i-1},t'_{i}]} \right] + \bbE\left[ \log  Z_{[t_L,N]}\right]
+\sum_{i=1}^L \log K(t_i-t'_i).
\end{equation} 
One has from Lemma \ref{finitevol}
for all $i>1$
$$\bbE Z_{[t_{i-1},t'_{i}]}\le (t'_{i}-t_{i-1})\tf(\gb, h) + h,$$ 
the extra $h$ term is there because the definition of $Z_[a,b]$ \eqref{defznab} takes into-account the environment at the starting point (

 and the fact that $h\le 1$, it follows that
\begin{multline}
 \sum_{i=1}^L \bbE\left[\log Z_{[t_{i-1},t'_{i}] }\right] + \bbE\left[ \log  Z_{[t_L,N]}\right] \\
 \le \left(\sum_{i=1}^L(t'_i-t_{i-1})+ (N-t_L)\right)\tf+ L h
 \le N \tf/2+ L.
\end{multline}
is not exactly the partition function (

Regarding the last term in \eqref{sketletor}, using Jensen's inequality for the function $x\mapsto x^{\zeta}$, we have, choosing $\delta>0$ sufficiently small, for $\gep$ sufficiently small 
\begin{equation}
-\sum_{i=1}^L \log K(t_i-t'_i)\ge (1-\delta) \sum_{i=1}^L (t_i-t'_i)^{\zeta}\ge (1-\delta) 2^{-\zeta} L l^{\zeta}\ge \frac{1}{2}L l^{\zeta},
\end{equation}
which ends the proof.
\end{proof}

\begin{lemma}\label{ohm}
There exists a constant  $C$  such that for $N$ sufficiently large, 
for all $L\in \{1,\dots,\gep N\}$,
\begin{equation}
\bbP\left(  \max_{(\bt',\bt) \in \cT(L)} \Big( \log Z^{\go}_N (A_{(\bt',\bt)})- \bbE[ \log Z^{\go}_N (A_{(\bt',\bt)})]\Big)\ge  C\gb N\sqrt{\frac{\gep \log l}{l}} \right)\le 
 \frac{1}{N^3}.
\end{equation}
\end{lemma}

\begin{proof}
From Lemma \ref{concentration} and  a standard union bounds one has for any $u$ 

\begin{equation}
\bbP\left(  \max_{(\bt',\bt) \in \cT(L)} \left( \log Z^{\go}_N (A_{(\bt',\bt)})- \bbE[ \log Z^{\go}_N (A_{(\bt',\bt)})]\right) \ge u\right)\le  (\#\cT)C_1  \exp\left(-\frac{u^2}{C_2\gb^2 N\gep}\right)
\end{equation}
Using Lemma \ref{one} for the value of $u:=C\gb N\sqrt{\frac{\gep \log l}{l}}$ one can conclude provided that $C$ is chosen sufficiently large.
\end{proof}

\begin{proof}[Proof of Proposition \ref{keypoint}]
Using Lemma \ref{ohm} and Lemma \ref{shanti},
one has almost surely for all large $N$, for all $L\le \gep N$
\begin{equation}
 \frac{1}{N}\max_{(\bt',\bt) \in \cT(L)}  \log Z^{\go}_N (A_{(\bt',\bt)})\le \frac{1}{2}\tf(\gb,h)+l^{-1}+C\gb \sqrt{\frac{\gep \log l}{l}}-\frac{1}{2} l^{\zeta-1}.
\end{equation}
Combining this with \eqref{cramnik} we obtain
\begin{equation}
\frac 1 N\log Z_N^\go(\cA^\gep)\le \frac{\log N} N + \max_{l\ge \gep^{-1}} \left( \frac{\log \# \cT(L)}{N}+\frac{1}{2}\tf(\gb,h)+ l^{-1}+  C\gb  \sqrt{\frac{\gep \log l}{l}}-\frac{1}{2} l^{\zeta-1}\right).
\end{equation}
The terms $\frac{\log N} N$ and  $\frac{\log \# \cT(L)}{N}$ can be neglected if $l$ is sufficiently large (i.e. $\gep$ is sufficiently small) and $l^{\zeta-1}/2$ is replaced by $l^{\zeta-1}/4$.
\end{proof}

\section{Proof of Theorem \ref{mainres2}: rounding for $\zeta< 1/2$}\label{seclwb}

The idea to find an upper bound on the free-energy is close to the one in \cite{cf:GT05}. The main difference is that here, we must combine the argument with the finite volume criterion given by Lemma \ref{finitevol} to get a result.
We use the fact that $\go$ is Gaussian in  the following  way:

\begin{lemma}\label{gaussianlemma}
For any $N$  if $\go$ are IID Gaussian variables then the 
sequence 
$$\left(\go_x-\frac{1}{N}\sum_{n=1}^N \go_n\right)_{x=1}^N$$ 
is independent of  
$\sum_{n=1}^N \go_n$. 
\end{lemma}
With this observation, we see that changing the value of $h$ by an amount $\delta$ is in fact equivalent to changing the empirical mean of the $\go$ by an amount $\delta \gb^{-1}$.

\medskip

In a first step we try to control the expectation of the free-energy for a typical value of $\sum_{n=1}^N \go_n$.

\begin{proposition}\label{bound}
There exists a constant $C$ such that 
for all $N$ sufficiently large  and all $u$,
\begin{equation}\label{dabound}
\bbE\left[\log Z^{\gb,h_c(\gb),\go}_N \ \big| \ \sum_{n=1}^N  \go_n\ge u \sqrt{N} \right]
\le C N^{\zeta}(1+|u|^{\zeta})e^{\zeta u^2/2}+\gb^2.
\end{equation}
\end{proposition}

This will be done using the finite volume criterion of Lemma \ref{finitevol}: if \eqref{dabound} does not hold, one can find a strategy which gives a positive free-energy for $h=h_c(\gb)$ and hence yields a contradiction.
Then the idea is to integrate this bound over all values of $u$ to obtain a bound for $\bbE\left[\log Z^{\gb,h,\go}_N\right]$.
Of course the bound will be a good one only if $N$ is wisely chosen. We can finally conclude using the finite volume criterion Lemma \ref{finitevol}.

\begin{proof}[Proof of Theorem \ref{mainres2}]

Now for $h=h_c(\gb)+v$ one sets $N:=\gb^2 v^{-2}$ (assuming that we have chosen $v$ such that $N$ is an integer).
One has 

\begin{equation}\begin{split}
\bbE\left[ \log Z^{\gb,h,\go}_N\right]&= \int \frac{1}{\sqrt{2\pi}}\exp\left(-u^2/2\right)\bbE\left[\log Z^{\gb,h,\go}_N \ | \ \sum_{n=1}^N  \go_n= u \sqrt{N} \right]\dd u\\
&=\int \frac{1}{\sqrt{2\pi}}\exp\left(-\frac{(u-\gb^{-1} v\sqrt{N})^2}{2}\right)\bbE\left[\log Z^{\gb,h_c(\gb),\go}_N \ | \ \sum_{n=1}^N  \go_n= u \sqrt{N} \right]\dd u.
\end{split}
\end{equation}

Using Proposition \ref{bound} we have the following inequality provided that $v$ is sufficiently small (in which case the $\gb^2$ can be neglected)

\begin{equation}
\bbE\left[ \log Z^{\gb,h,\go}_N\right]\le C N^{\zeta} \int \frac{1}{\sqrt{2\pi}}(1+|u|^{\zeta}) \exp\left(\frac{\zeta u^2-(u-1)^2}{2}\right)\dd u\le C' N^{\zeta}.
\end{equation}
Hence, using Lemma \ref{finitevol}, we obtain 

\begin{equation}
\tf(\gb,h)\le C'' N^{\zeta-1}= C'' (v\gb^{-1})^{2(1-\zeta)}.
\end{equation}

\end{proof}
\begin{proof}[Proof of Proposition \ref{bound}]
One can assume $u\ge 1$ without loss of generality.
Set $$M:= u\exp\left(u^2/2\right).$$
Let $X_0$ be the smallest integer such that
\begin{equation}
 \sum_{n=X_0 N+1}^{(X_0+1)N}  \go_x\ge u \sqrt{N}.
\end{equation}
Then we obtain a lower bound on $Z_{NM}$ by deciding to visit the stretch $[X_0N,(X_0+1)N]$ if $X_0\le M- 2$ 
and to do only a long excursion in the other case (recall \eqref{znab}) (see Figure \ref{factorx}):

\begin{figure}[hlt]
 \epsfxsize =11.5 cm
 \begin{center}
 \psfrag{O}{$0$}
  \psfrag{NM}{$NM$}
  \psfrag{X0N}{\tiny {$X_0N$}}
    \psfrag{X01N}{\tiny {$(X_0+1)N$}}
\epsfbox{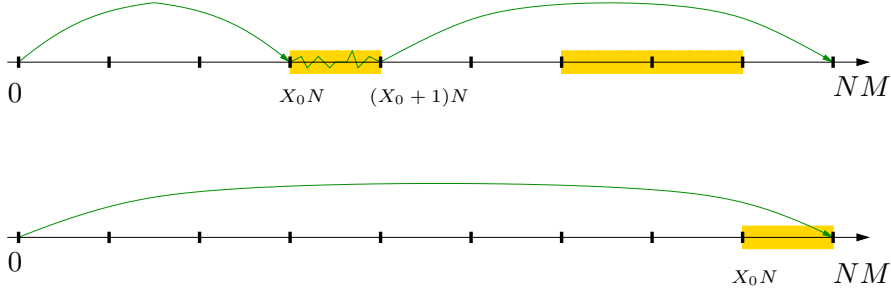}
 \end{center}
 \caption{\label{factorx}  
Here we present our strategy to obtain a lower bound on the partition function $Z_{NM}$. The yellow segments are those which are such that the empirical mean of $\go$ is larger than $u N^{-1/2}$.
By definition $[X_0N,(X_0+1)N]$ is the first of these segments. We allow pinning only on the segment  $[X_0N,(X_0+1)N]$ and only if $X_0\le M-2$. }
 \end{figure}

\begin{equation}
 Z^\go_{MN}\ge \begin{cases} K(X_0N)  Z^\go_{[X_0N,(X_0+1)N]}K((M-X_0+1)N)e^{\gb \go_{NM}+h_c(\gb)}, &\text{ if } X_0\le (M-2), \\
 K(MN)e^{\gb \go_{NM}+h_c(\gb)} \quad &\text{ if } X_0\ge M-2.
\end{cases}
\end{equation}
Taking the expectation one obtains, by translation invariance

\begin{multline}
\bbE \left[\log \left( K(X_0N) \log Z^\go_{[X_0N,(X_0+1)N]}K((M-X_0+1)N)\right) \ | \ X_0\le (M-2) \right]\\
\ge -2(MN)^{\zeta}+h_c(\gb)+\bbE\left[\log Z^{\gb,h_c(\gb),\go}_N \ | \ \sum_{n=1}^N  \go_x\ge u \sqrt{N} \right].
\end{multline}
We also have (as $\go_{MN}$ is independent of the event $\{X_0\le M-2\}$ its conditional mean is zero)
\begin{equation}
 \bbE\left[ \log \left[K(MN)e^{\gb \go_{NM}+h_c(\gb)}\right] \ | \  X_0\ge M-2 \right]= \log K(MN)+h_c(\gb).
\end{equation}
And hence (recall that $h_c(\gb)\ge -\gb^2/2$ for Gaussian environments),

\begin{equation}
\bbE[\log Z^\go_{MN}]\ge \bbE\left[\log Z^{\gb,h_c(\gb),\go}_N \ | \ \sum_{n=1}^N  \go_x\ge u \sqrt{N} \right] \bbP\left[ X_0\le M-2 \right]-2 (MN)^{\zeta}-\gb^2.
\end{equation}
By standard estimates on Gaussian tails, there exists a constant $c>0$ such that  
$$\forall u >1, \quad \bbP\left[ \sum_{n=1}^N  \go_x\ge u \sqrt{N} \right]\ge \frac c u e^{-u^2/2},$$  and hence, using the definition of $M$
we have
$$\bbP\left[ X_0\le M-2 \right]>c',$$ 
for some positive constant $c'$.
This implies  (recall Lemma \ref{finitevol} and that $\tf(\gb,h_c(\gb))=0$) that there exists $c''>0$ such that

\begin{equation}
0\ge \bbE[\log Z^\go_{MN}]\ge c' \left( \bbE\left[\log Z^{\gb,h_c(\gb),\go}_N \ | \ \sum_{n=1}^N  \go_x\ge u \sqrt{N} \right]\right) -c'' u^{\zeta} e^{\zeta u^2/2} N^\zeta-\gb^2.
\end{equation}
The above inequality is in fact only valid if one assumes that 
$$
 \bbE\left[\log Z^{\gb,h_c(\gb),\go}_N \ | \ \sum_{n=1}^N  \go_x\ge u \sqrt{N} \right]\ge 0,$$
 but if this is not the case there is nothing to prove.
\end{proof}

\section{Proof of Theorem \ref{mainres3}: rounding for $\zeta= 1/2$}

The case for $\zeta=1/2$ is a bit more complicated. 
Assume that 
\begin{equation}\label{defc0}
\lim_{h\to h_c(\gb)+} \partial_h\tf(\gb,h)=c_0>0,
\end{equation}
and let us derive a contradiction.
Fist,  we prove that the contact fraction at the critical point, if well defined, cannot be equal to $c_0$ as there is always a positive probability for 
the polymer to have a very small contact fraction. 

\begin{lemma}\label{criticool}
The following three statements hold
\begin{itemize}
\item [(i)]
For all $\gep>0$, one has 
\begin{equation}
\limsup_{N\to \infty} \bbE\left[\bP^{\gb,h_c(\gb),\go}_N\left(\cB^\gep \right) \right]<1.
\end{equation} 
\item [(ii)]For any $u>c_0$ one has 
\begin{equation}
\lim_{N\to \infty} \bbE\left[\bP^{\gb,h_c(\gb),\go}_N\left(\cB^u \right) \right]=0.
\end{equation} 
\item [(iii)]One has 
\begin{equation}
\limsup_{N\to \infty} \frac 1 N \bbE\left[\bE^{\gb,h_c(\gb),\go}_N\left(\sum_{n=1}^N\delta_N \right) \right]<c_0.
\end{equation} 
\end{itemize}
\end{lemma}

\begin{proof}
Point $(iii)$ is a simple consequence of the two first point as 
\begin{equation}
\frac 1 N \bbE\left[\bE^{\gb,h_c(\gb),\go}_N\left(\sum_{n=1}^N\delta_N \right) \right]=\frac 1 N\int_0^1 \bbE\left[\bP^{\gb,h_c(\gb),\go}_N\left(\cB^u \right) \right]\dd u.
\end{equation}
Point $(ii)$ is rather easy to prove:
Assume that for $u>c_0$ and for some $\delta>0$ one has

\begin{equation}
\bbP\left[\bP^{\gb,h_c(\gb),\go}_N\left(\cB^u \right)>\delta\right]>\delta,
\end{equation}
for infinitely many $N$.

We note that if 
$$\bP^{\gb,h_c(\gb),\go}_N\left(\cB^u \right)>\delta,$$ then
\begin{equation}
Z^{\gb,h_c(\gb),\go}_N(\cB^u)\ge \delta Z^{\gb,h_c(\gb),\go}_N \ge \delta K (N)e^{\gb\go_N+h_c(\gb)},
\end{equation}
where the last inequality is just obtained by considering renewal trajectories with only one contact.
Hence, for every $h>h_c(\gb)$ we have
\begin{equation}
Z^{\gb,h,\go}_N\ge Z^{\gb,h,\go}_N(\cB^u)\ge \delta e^{Nu(h-h_c)} K(N)e^{\gb\go_N+h_c(\gb)}.
\end{equation}
This implies (as we know that the limit exists and is non-random) that for every $h>h_c(\gb)$
\begin{equation}
\lim_{N\to \infty} \frac{1}{N}\log  Z^{\gb,h,\go}_N \ge u(h-h_c(\gb))
\end{equation}
which contradicts assumption \eqref{defc0} for small $h$.

\medskip

To prove $(i)$ let us assume that
\begin{equation}\label{opc}
\lim_{N\to \infty} \bbE\left[\bP^{\gb,h_c(\gb),\go}_N\left(\cB^\gep  \right) \right]=1,
\end{equation}
(or that it occurs along a subsequence) and derive a contradiction from it.
Set 
\begin{equation}
 f_N(u):=\bbE\left[\bP^{\gb,h_c(\gb),\go}_N\left(\cB^\gep \right)  \ | \ \sum_{x=1}^{N-1} \go_x= u\sqrt{N-1} \right].
\end{equation}
We have 
\begin{equation}
 \bbE\left[\bP^{\gb,h_c(\gb),\go}_N\left(\cB^\gep \right) \right]=\int\frac{1}{\sqrt{2\pi}}\exp\left(-u^2/2\right)f_N(u)\dd u.
\end{equation}
As $f_N(u)$ is an increasing function of $u$ this implies that for all $u\in \bbR$
\begin{equation}
\lim_{N\to\infty} f_N(u)=1.
\end{equation}
Fix $u=-10 \gep^{-1}$ and let $N$ be sufficiently large so that $f_N(u)\ge 3/4$.
Then necessarily 
\begin{equation}\label{sixquatorz}
\bbP\left(  \bP^{\gb,h_c(\gb),\go}_N\left(\cB^\gep \right)\ge 1/2   \  | \ \sum_{x=1}^{N-1} \go_x= u\sqrt{N-1} \right)\ge 1/2.
\end{equation}
Note that $\bP^{\gb,h_c(\gb),\go}_N\left(\cB^\gep \right)\ge 1/2$ implies in particular that 
$$ Z^{\gb,h_c(\gb),\go}_N(\cB^\gep)\ge Z^{\gb,h_c(\gb),\go}_N((\cB^\gep)^c)\ge K(N)e^{\gb\go_N+h_c(\gb)}.$$
And hence \eqref{sixquatorz}
\begin{equation}
\bbP\left( Z^{\gb,h_c(\gb),\go}_N(\cB^\gep)\ge K(N)e^{\gb \go_N+h_c(\gb)}  \  | \ \sum_{x=1}^{N-1} \go_x= u\sqrt{N-1} \right)\ge 1/2.
\end{equation}

From Lemma \ref{gaussianlemma}, replacing $u$ by $v$ in the conditioning is equivalent to replacing $\go_n$ by $\go_n+(v-u)(N-1)^{-1/2}$
for $n\in \{1,\dots,N-1\}$.
Hence for $v\ge u$ we have
\begin{equation}\label{decompo}
\bbP\left( Z^{\gb,h_c(\gb),\go}_N(\cB^\gep)\ge K(N) e^{\gep(v-u) \sqrt{N-1}+\gb\go_N+h_c(\gb)} \  | \ \sum_{x=1}^{N-1} \go_x= v\sqrt{N-1} \right)\ge 1/2.
\end{equation}
This implies that for any $v$ (this is obvious for $v\le u$)
\begin{equation}
\bbP\left( Z^{\gb,h_c(\gb),\go}_N\ge  K(N) e^{\gep(v-u) \sqrt{N-1}+\gb\go_N+h_c(\gb)}  \  | \ \sum_{x=1}^{N-1} \go_x= v\sqrt{N-1} \right)\ge 1/2.
\end{equation}
Hence, 
using the obvious bound $Z^{\gb,h_c(\gb),\go}_N\ge K(N)e^{\gb\go_N+h_c(\gb)}$, one obtains
\begin{equation}
\bbE\left[ \log Z^{\gb,h_c(\gb),\go}_N   \  | \ \sum_{x=1}^{N-1} \go_x= v\sqrt{N-1} \right]\\
\ge 
\log K(N)+h_c(\gb)+\frac{1}{2}\gep(v-u) \sqrt{N-1}.
\end{equation}

Hence integrating over $v$ one obtains (recall the value we have chosen for $u$)
\begin{multline}
\bbE\left[ \log Z^{\gb,h_c(\gb),\go}_N\right]
\ge \log K(N)+h_c(\gb) + \frac{1}{2}\frac{1}{\sqrt{2\pi}}\int \gep(v-u)\sqrt{N-1}e^{-\frac{v^2}{2}}\dd v \\
=\log K(N)+h_c(\gb) - \frac {\gep u } {2} \sqrt{N-1} = \log K(N)+h_c(\gb)+5\sqrt{N-1}>0.
\end{multline}
This contradicts the fact that the free-energy is zero.
\end{proof}

Then we can conclude by exhibiting a finite volume bound similar to those of Lemma \ref{finitevol} for the 
free energy derivative.

\begin{lemma}\label{finitevol2}
For $K$ log-convex, for any $N$ and $h$
\begin{equation}
\frac{1}{N}\bbE\left[\bE^{\gb,h,\go}_N\left(\sum_{n=1}^N\delta_N \right) \right]\ge \partial_h\tf(\gb,h).
\end{equation}
\end{lemma}

\begin{proof}
This is a simple consequence of the FKG inequality, as the number of contacts is an increasing function.
For $M\ge 1$ on has 
\begin{multline}
\bE^{\gb,h,\go}_{MN}\left[\sum_{n=1}^{NM}\delta_{n} \right] 
\ge 
\bE^{\gb,h,\go}_{MN}\left[\sum_{n=1}^{NM}\delta_{n} \ | \ \delta_{iN}=1, \  \forall i\in \{1,\dots, M-1\}\right]
\\ = \sum_{i=0}^{M-1} \bE^{\gb,h,\theta^{iN} \go}_{N}\left[\sum_{n=1}^{N}\delta_{n}\right],
\end{multline}
and hence taking the average 
\begin{equation}
\frac{1}{NM}\bbE\left[\bE^{\gb,h,\go}_{MN}\left[\sum_{n=1}^{NM}\delta_{n} \right] \right]\le \frac{1}{N}\bbE\left[\bE^{\gb,h,\go}_{N}\left[\sum_{n=1}^{N}\delta_{n} \right]\right].
\end{equation}
The result follows by taking $M$ to infinity.
\end{proof}

\begin{proof}[Proof of Theorem \ref{mainres3}]

For a fixed $N$, 
$$ h\mapsto  \frac{1}{N}\bbE\left[\bE^{\gb,h,\go}_{N}\left[\sum_{n=1}^{N}\delta_{n} \right]\right]. $$
is a continuous function. Hence from \eqref{criticool} one can find $N$ sufficiently large and $h>h_c$ such that 
\begin{equation}
 \frac{1}{N}\bbE\left[\bE^{\gb,h,\go}_{N}\left[\sum_{n=1}^{N}\delta_{n} \right]\right]< c_0.
 \end{equation}
 By Lemma \ref{finitevol2}, this implies that $\partial_h \tf (\gb,h)<c_0$ which yields a contradiction.
 Hence one must have a smooth transition.

\end{proof}

\begin{remark}
In fact the proof in this section yields a non trivial result for  $\zeta<1/2$: when $K$ is $\log$-convex one has  
\begin{equation}
 \lim_{N\to \infty} \frac{1}{N}\bbE\left[\bE^{\gb,h_c(\gb),\go}_{N}\left[\sum_{n=1}^{N}\delta_{n} \right]\right]=\lim_{h\to h_c(\gb)+} \partial_h\tf(\gb,h).
\end{equation}
In other words the contact fraction at the critical point is equal to the right-derivative of the free-energy.
\end{remark}

{\bf Acknowledgement:} The author is grateful to G. Giacomin for enlightening discussion on the subject and to N. Torri for providing access to reference \cite{cf:T14}.

\end{document}